\documentclass[conference]{IEEEtran}
\ifCLASSINFOpdf
\else
\fi
\usepackage{times}

\usepackage{soul}
\usepackage{url}

\usepackage[hidelinks]{hyperref}
\usepackage[utf8]{inputenc}
\usepackage[small]{caption}
\usepackage{graphicx}
\usepackage{amsmath}
\usepackage{booktabs}
\usepackage{subcaption}
\usepackage{cleveref}
\usepackage{hyperref}
\usepackage{breakurl}
\usepackage{bbm}
\usepackage{amsthm}
\usepackage{enumitem}
\usepackage{multirow}
\usepackage{amsfonts}
\usepackage{microtype}
\usepackage{tabularx}
\usepackage{siunitx}
\newtheorem{theorem}{Theorem}
\newcommand{\tabincell}[2]{\begin{tabular}{@{}#1@{}}#2\end{tabular}} 
\usepackage{multirow} 
\DeclareMathOperator*{\argmax}{argmax}
\usepackage{graphicx} 
\usepackage{soul}
\usepackage{xcolor}
\theoremstyle{definition}

\usepackage[ruled,linesnumbered]{algorithm2e} 

\usepackage{float}
\usepackage{hhline}
\usepackage{caption}
\usepackage{authblk}

\begin{document}
%
\title{Impression Allocation and Policy Search\\in Display Advertising}




\author[1]{Di Wu\thanks{B.B@university.edu}}
\author[1]{Cheng Chen}
\author[1]{Xiujun Chen}
\author[2]{Junwei Pan}
\author[1]{Xun Yang}
\author[1]{Qing Tan}
\author[1]{Jian Xu}
\author[1]{Kuang-Chih Lee}

\affil[1]{Alibaba Group}
\affil[2]{Yahoo Research}

\affil[ ]{\{wudi.xjtu,chencheng1022,kongkongchen,pandevirus,yangxun1412,granttanqing\}@gmail.com}
\affil[ ]{\{xiyu.xj,kuang-chih.lee\}@alibaba-inc.com}


%


\maketitle

\begin{abstract}
In online display advertising, guaranteed contracts and real-time bidding (RTB) are two major ways to sell impressions for a publisher. For large publishers, simultaneously selling impressions through both guaranteed contracts and in-house RTB has become a popular choice. Generally speaking, a publisher needs to derive an impression allocation strategy between guaranteed contracts and RTB to maximize its overall \textit{outcome} (\textit{e.g.}, revenue and/or impression quality). However, deriving the optimal strategy is not a trivial task, \textit{e.g.}, the strategy should encourage incentive compatibility in RTB and tackle common challenges in real-world applications such as unstable traffic patterns (\textit{e.g.}, impression volume and bid landscape changing). In this paper, we formulate impression allocation as an auction problem where each guaranteed contract submits virtual bids for individual impressions. With this formulation, we derive the optimal bidding functions for the guaranteed contracts, which result in the optimal impression allocation. In order to address the unstable traffic pattern challenge and achieve the optimal overall \textit{outcome}, we propose a multi-agent reinforcement learning method to adjust the bids from each guaranteed contract, which is simple, converging efficiently and scalable. The experiments conducted on real-world datasets demonstrate the effectiveness of our method.
\end{abstract}


%
\IEEEpeerreviewmaketitle

\section{Introduction}

Online display advertising has become one of the most influential business, with \$59.8 billion revenue in FY 2019 in US alone \cite{iab-fy-2019}. 
Typically when a user visits a publisher, e.g., a news website, there would be one or more ad impression opportunities generated in real time. Advertisers are able to acquire these opportunities to display their ads at certain cost and these cost eventually become the revenue of the publisher.

For a publisher, there are two major ways to sell impressions. The first one is through \emph{guaranteed contracts} (also referred as \emph{guaranteed delivery} \cite{chen2014dynamic}). A guaranteed contract is an agreement between an advertiser and a publisher by negotiating directly or by going through a programmatic guaranteed mechanism~\cite{chencombining}. The contract usually specifies the contract payment amount, the campaign duration and the desired number of ad impressions. The advertiser typically makes the payment before the ad delivery starts and the publisher guarantees the desired number of ad impressions. The publisher is also responsible for any shortfall in the number of impressions delivered. A penalty is usually incurred based on the volume of under-delivery.

The second way to sell impressions is through \emph{real-time bidding} (RTB). RTB allows advertisers to bid in real time for impressions and does not guarantee the impression volume for any advertiser~\cite{yuan2013real}. In this paper, we focus on the second price auction~\cite{edelman2007internet}. For each impression opportunity, the advertiser who offers the highest bid wins the opportunity to display her ad. The cost of the winner is the second highest bid in the auction.


Despite the increasing popularity of RTB, there is still large amount of the online display advertising revenue generated from guaranteed contracts. 
For large publishers, simultaneously selling impressions through both guaranteed contracts and in-house RTB has become a popular choice. These publishers need to derive impression allocation strategies between guaranteed contracts and RTB. There are two main considerations when publishers strategically allocate their impressions: 
\begin{enumerate}[align=right,leftmargin=0.17in]
	\item \textbf{Maximizing the overall outcome}: The overall \textit{outcome} we consider in this paper consists of both revenue and contract impression quality (a formal definition can be found in Section \ref{section_yo}). Revenue represents the short-term value to the publisher. It associates with the revenue from guaranteed contracts, the revenue from RTB, and the contract violation penalties. Impression quality (\textit{e.g.}, click-through rate) of the guaranteed contracts reflects the long-term revenue to the publisher since advertisers of guaranteed contracts are more and more concerned with the impression qualities. 
	
	\item \textbf{Maintaining incentive compatibility in RTB}: The publisher also needs to take care of the auction mechanism of its in-house RTB. An incentive compatible auction mechanism\footnote{A mechanism is called incentive-compatible if every participant can achieve the best \textit{outcome} to themselves just by acting according to their true preferences.} is essential to facilitate truthful bidding and maximize the efficiency.
\end{enumerate}



Besides the considerations above, what makes it even more challenging is the unstable traffic pattern in real-world applications. Strategies derived based on a known collection of impressions \cite{balseiro2014yield,ghosh2009bidding,jauvion2018optimal} or stochastic arrival models \cite{vee2010optimal,feldman2010online,devanur2009adwords,devanur2011near} may be inferior in such scenarios. First, the traffic volume is vulnerable to unexpected changes such as holiday sale events. Second, usually concurrent with the traffic volume changes, the market price distribution of the impressions can also deviate from the empirical one. Finally, the unpredictable advertiser behaviors in RTB including modifying budget, bid, and target audience can also make the traffic pattern complicated and dynamic \cite{wu2018budget}. 

To derive an optimal impression allocation strategy that takes into account the above-mentioned considerations and challenges, we propose to analyze the problem from another perspective. Since the allocation is non-trivial and the environment is highly dynamic, can the guaranteed contracts also participate in the real-time auctions so that they can also enjoy the liquidity and the impressions can be fully auctioned? More specifically, can each guaranteed contract be treated as a bidding agent which is able to submit bids for individual impressions and the impression allocation is based on the submitted bids from both guaranteed contracts and RTB? We will show that such a setup can actually lead us to the optimal impression allocation strategy.
For each impression, the optimal bids of each contract bidding agent $j$ are determined by a critical parameter $\alpha_j$, and the impression can be optimally allocated based on an auction mechanism with the bids from RTB and contract agents. It can be proved that the proposed mechanism results in the theoretically maximum \textit{outcome} (more details see Section \ref{section_yo}). Meanwhile, the incentive compatibility, which is important for auction efficiency, is also maintained.


However, although the impressions can be allocated in such scheme, the optimal $\alpha_j$ can only be obtained over a complete impression set. Hence, an $\alpha_j$ adjustment policy is a necessity in real-world application, which aims to continuously adjust $\alpha_j$ to the optimal one under the current state (\textit{i.e.}, indicators about contracts' fulfillment status and RTB information). Meanwhile, since each agent $j$ has its own parameter $\alpha_j$, and all the agents have a common goal of \textit{outcome} maximization, it is intuitive to apply multi-agent reinforcement learning (MARL) to model this process. However, MARL method is faced with some common challenges in industrial scenario, such as model complexity, scalability and low converging efficiency.
Thus, we carefully design our modeling process and present a simple, scalable, quickly converging MARL method.

To evaluate the effectiveness of our method, we conducted experiments on large-scale real-world datasets. Compared with other methods, we observed substantial improvements on impression allocation results. Meanwhile, the converging efficiency, scalability of our MARL method are also verified empirically. 
Our main contributions can be summarized as follows:
\begin{enumerate}[align=right,leftmargin=0.15in]
	\item We propose an optimal impression allocation strategy in display advertising with both guaranteed contracts and in-house RTB. The proposed strategy maintains incentive compatibility in RTB.
	\item We devise an industrial applicable MARL method to dynamically optimize impression allocation \textit{outcome}, which is simple, scalable and quickly converging.
	\item Empirical studies on real-world industrial dataset demonstrate the scalability and efficiency of our MARL method.
\end{enumerate}

The rest of this paper is organized as follows. The optimal impression allocation strategy is derived associated with optimal bidding function in Section \ref{section_yo}. In Section \ref{MARLIA}, we present an efficient MARL method to learn the crucial parameter in bidding function. Empirical study is shown in Section \ref{section:exp}, followed by the related work in Section \ref{relatedwork}. We conclude the paper in Section \ref{conclusion}.

\section{Optimal Impression Allocation} \label{section_yo}

\begin{figure}
	\centering
	\includegraphics[width=0.45\textwidth]{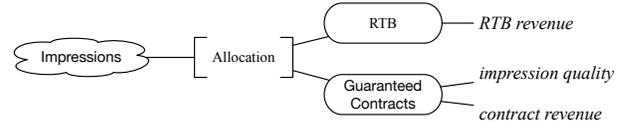}
	\caption{Impression allocation process of a publisher.}
	\label{fg:hy_method_framework}
\end{figure}

For the publishers who sell impressions through both guaranteed contracts and in-house RTB, one of their impression allocation motivations is to maximize the total revenue from both contracts and RTB. Meanwhile, the impression quality of the contracts can affect the satisfaction of the contract advertisers and therefore affect the long-term revenue. 
As illustrated by Fig. \ref{fg:hy_method_framework}, 
the ultimate goal of impression allocation is to maximize the overall \textit{outcome}, \textit{i.e.}, simultaneously maximize RTB revenue
, contract revenue and contract impression quality. Besides, whether the RTB part is incentive compatible is also a critical issue to be considered. 

\subsection{Problem Formulation} \label{section_pf}

Suppose there are $n$ impressions indexed by $i$ to be allocated by the publisher. On the one hand, suppose that there are $m$ guaranteed contracts indexed by $j$ to be served. For each contract $j$, let $d_j$ be the demand impression volume and $c_j$ be the unit price of each impression, so the prepaid contract revenue is $c_jd_j$. Suppose the contract violation penalty of contract $j$ for each undelivered impression is $p_j$, that is, if the number of impressions served to contract $j$ is fewer than $d_j$, the publisher will have to be responsible for the under-delivery and the penalty is $p_j$ for each undelivered impression. On the other hand, for each impression $i$, RTB will also provide a list of bids, of which we are mostly interested in the first and second highest bids $\mathbbm{b}_{i1}$ and $\mathbbm{b}_{i2}$ under the second price auction mechanism. If the impression is allocated to RTB, then the publisher will earn $\mathbbm{b}_{i2}$. Let $x_{ij}$ be the binary indicator whether an impression $i$ is allocated to contract $j$ and obviously if $\sum_j{x_{ij}=0}$ then the impression is allocated to RTB. Then we are able to derive revenue from both guaranteed contracts and RTB. 
Specifically, the revenue from guaranteed contracts is $R_{GC} = \Sigma_j c_jd_j - \Sigma_j p_jy_j$, where $y_j=d_j-\Sigma_i x_{ij}$ is the impression under-delivery amount and the revenue from RTB is $R_{RTB}=\Sigma_i (1-\Sigma_j x_{ij})\mathbbm{b}_{i2}$. For contract business, a publisher should consider the quality of the impressions allocated to contracts, since the result of poor impression quality (\textit{e.g.}, low CTR) would drastically affect the future investments on contract ads. 
In industrial application, it is common to introduce a trade-off parameter between the units of impression quality and revenue \cite{balseiro2014yield}.
Thus, we let $q_{ij}$ be the impression $i$'s quality for contract $j$ and $\lambda_j$ be the quality weight\footnote{The quality weight can be also considered as a scaling parameter that converts quality into money.}, then the total contract impression quality is $Q_{GC} =\Sigma_{ij} \lambda_jx_{ij}q_{ij}$.

Our goal of impression allocation is to maximize the overall \textit{outcome}, \textit{i.e.}, the sum of $R_{GC}, R_{RTB}$, and quality $Q_{GC}$. One may curious about the reason why we do not consider the quality of the RTB part, this is mainly because the RTB part is fully auctioned and the impression quality is already considered in the bids of each advertiser in RTB, \textit{e.g.}, if an advertiser considers click as the quality of an impression, he or she can only bid with a click price. Hence, we formulate the problem of \emph{optimal impression allocation for outcome maximization} as the following linear programming problem:

\begin{equation} \tag{LP1} \label{lp1}
	\begin{aligned}
		\underset{x_{ij},y_j}{\text{maximize}}  & & R_{GC} + R_{RTB} + Q_{GC} & & & \\
		\text{s.t.} & & \Sigma_i x_{ij} + y_{j} = d_j,\quad &&&\forall j, \nonumber\\
		& & \Sigma_j x_{ij} \le 1,\quad &&& \forall i, \nonumber\\
		& & x_{ij} \ge 0,\quad &&& \forall i,j, \nonumber\\
		& & y_{j} \ge 0,\quad &&& \forall j. \nonumber
	\end{aligned}    
\end{equation}

\subsection{The Optimal Allocation Strategy} \label{section_oas}
As demonstrated above, the optimal impression allocation can be obtained by solving \eqref{lp1} if the impressions are completely foreseen. However, the impression is unknown until the end of the day and almost impossible to predict since the unstable traffic pattern. Therefore we propose to solve the problem from a different perspective. That is, guaranteed contracts are regarded as bidding agents\footnote{Different from the bidding agents in RTB, the contract bidding agents do not have budget, and the delivered bids are only used for the impression allocation.}, whose bids are decided by the publisher, and participate in the real-time auctions together with the RTB ads. The impression allocation is finally decided by the auction results of RTB and contracts. In this section, we prove that such a setup can actually lead us to
the optimal allocation strategy and derive the optimal bidding functions of contract bidding agents.

\begin{theorem} \label{theorem}
	Suppose for each impression $i$, every contract $j$ submits a bid $b_{ij}$. Let $k=\argmax_j b_{ij}$. Consider the allocation strategy  that allocates impression $i$ to contract $k$ if $b_{ik} > \mathbbm{b}_{i2}$ and otherwise to RTB. The strategy results in the optimal solution of outcome maximization problem defined by \eqref{lp1} if the bid from each contract is:
	\label{algo:prove}
	\begin{equation} \label{eq:optimal_f}
		b_{ij} = \lambda_j\ q_{ij} + \alpha^{*}_{j},\ \forall i,j,
	\end{equation}

\noindent
where $\alpha^{*}_{j}$ is the optimal solution to the dual problem of outcome maximization problem (LP1) and $\alpha^{*}_{j} \leq p_j$ .



\end{theorem}

\begin{proof}
	The maximal revenue from RTB $\Sigma_i\mathbbm{b}_{i2}$, along with the total payment from contracts $\Sigma_jc_jd_j$ can be considered as constants. Then, $\text{maximize} \ \Sigma_{ij}(1-x_{ij})\mathbbm{b}_{i2} + \Sigma_jc_jd_j - \Sigma_j y_j p_j + \Sigma_{ij} \lambda_jx_{ij}q_{ij}$ can be simplified as $\text{maximize}\ -\Sigma_{ij}x_{ij}\mathbbm{b}_{i2} - \Sigma_j y_j p_j + \Sigma_{ij} \lambda_jx_{ij}q_{ij}$. Thus, the dual problem of \eqref{lp1} is as follows:
	
	\begin{align}
		\underset{\alpha_j,\beta_i}{\text{minimize}} & & \Sigma_{i}\beta_{i} -\Sigma_{j} \alpha_j d_{j} & & & \tag{LP2}\label{lp2} \\
		\textup{s.t.} &  & \lambda_j q_{ij} + \alpha_{j}  - \mathbbm{b}_{i2} \le \beta_i,  & & &  \forall i,j, \label{lp2_1} \\
		& & \alpha_j \le p_{j},  & & &  \forall j, \nonumber\\ 
		& & \beta_i \ge 0,  & & & \forall i. \nonumber
	\end{align}
	
	Suppose the optimal solution to \eqref{lp1} and \eqref{lp2} are $x^*_{ij}, y_j^*$ and $\beta^*_i, \alpha^*_j$ respectively. We denote impression $i$'s bid of contract $j$ as $b_{ij} = \lambda_j q_{ij} + \alpha_{j}^*$. According to the complementary slackness theorem,  we have
	
	\begin{align}
		x_{ij}^* \ (b_{ij} - \mathbbm{ 	b}_{i2} - \beta_i^*) &= 0,     & & \forall i,j, \label{dual1}\\ 
		(\Sigma_j {x_{ij}^*} -1) \ \beta_i^* &= 0, & & \forall i. \label{dual2}
	\end{align}
	
	\noindent
	$\forall i=1..n; j=1..m$:
	\begin{itemize}
		\item If $b_{ij} < \mathbbm{b}_{i2}$ then $b_{ij} - \mathbbm{b}_{i2} - \beta_{i}^* < 0$. Based on Eq. \eqref{dual1} we can infer that $x_{ij}^* = 0$. Therefore if $\forall j, b_{ij} < \mathbbm{b}_{i2}$ then  $\sum_j{x_{ij}^* = 0}$, which means impression $i$ is allocated to RTB in the optimal solution of \eqref{lp1}.
		
		\item If $b_{ij} > \mathbbm{b}_{i2}$, we can infer that $\beta_i^* > 0$ and $\sum_j{x_{ij}^* = 1}$ according to Eqs. \eqref{lp2_1},\eqref{dual2}, which means impression $i$ is allocated to guaranteed contracts. Let $x_{ik}^* > 0$, then we have $b_{ik} = \mathbbm{b}_{i2} + \beta_{i}^*$ based on Eq. \eqref{dual1}. Therefore $b_{ij} \le \mathbbm{b}_{i2} + \beta_{i}^* = b_{ik}$, \textit{i.e.}, $k = \arg\max_j {b_{ij}}$. This means in the optimal solution of \eqref{lp1}, the impression is allocated to the contract with highest bid.
	\end{itemize}
	
	In summary, for each impression $i$, if we bid with $b_{ij} = \lambda_j q_{ij} + \alpha_{j}^*$ for each contract $j$, and allocate impressions with the strategy demonstrated in Eq. \eqref{allocate}, then we achieve the the optimal allocation results of \eqref{lp1}:
	
	\begin{equation}\label{allocate}
		\text{allocate to }
		\begin{cases}
			\text{ RTB} &  {\text{if }  \forall j, b_{ij} < \mathbbm{b}_{i2}}, \\
			\text{ contract }k & {\text{otherwise}},
		\end{cases} 
	\end{equation}
	where $k = \arg\max_j {b_{ij}}$.
\end{proof}

\subsection{Incentive Compatibility of RTB}\label{IC_prove}

Since simultaneously selling impressions through both guaranteed contracts and in-house RTB has become a popular choice for large publishers, the incentive compatibility of RTB should be seriously considered. If the auction mechanism in RTB is incentive compatible, truthful bidding strategy is facilitated for the advertisers who participate in the auction, and it will result in a locally envy-free equilibrium which maximizes the efficiency \cite{edelman2007internet}. 

According to the allocation strategy presented in Section \ref{section_oas}, for each impression $i$, we simply compare the max bid from contracts with the second highest bid from RTB, and then allocate the impression to the one with the higher bid.
From the perspective of any RTB bidder, the social choice function is monotone in every bid submitted by a RTB bidder, \textit{i.e.}, the impression allocation result (\textit{True} or \textit{False}) is monotone with respect to her bids, and the critical value (payment) for the winning RTB bidder is still the second highest price $\mathbbm{b}_{i2}$. According to the Theorem 9.36 in \cite{nisan2007algorithmic}, the optimal allocation strategy introduced in Section \ref{section_oas} will keep the incentive compatibility of RTB.

\section{A Practical MARL Method} \label{MARLIA}
Recall that all the guaranteed contracts are considered as bidders (publisher proxies) to compete with RTB bidders in the auction. 
The optimal bidding function of contracts is given by Eq.~\eqref{eq:optimal_f}, and the remaining problem of impression allocation is the determination of the optimal bidding parameter $\alpha^*_j$. Since the traffic pattern is unstable and the impression set is unknown until the end of the day, it is impossible to obtain the optimal $\alpha_j^*$ by solving the dual problem of \eqref{lp1}. Thus, in the real-world scenario, we have to adjust the current (not optimal) $ \alpha_j$ according to the current state (\textit{e.g.}, contract fulfillment rate, cost per mille in RTB, etc.) to maximize the potential \textit{outcome}. 
In this section, we first formulate this $\alpha_j$ adjustment problem as a Markov Game, then we simplify the policy searching process based on an important property in our scenario, last we present an industry-applicable parameter adjustment approach via multi-agent reinforcement learning (MARL).  

\subsection{Formulation as A Markov Game} \label{markovGame}
%
%

We formulate the impression allocation process as an Markov Game, where there are $m$ contract agents and each agent $j \in [1..m]$ adjusts its $\alpha_j$ in order to achieve the global \textit{outcome} maximization. A Markov game is defined by a set of states $\mathcal{S}$ describing the status of impression allocation, a set of observations $\mathcal{O}_j$ indicating the observed information of the current state from the perspective of agent $j$, a set of actions $\mathcal{A}_1, ... , \mathcal{A}_m$ where $\mathcal{A}_j\subseteq \mathbb{R}$ represents the action space of agent $j$. At each time step $t$, each agent $j$ observes $o_{j,t}$ and then delivers an action $\delta_{j,t}$ to adjust $\alpha_{j,t}$ according to its policy $\pi_j:\mathcal{O}_j\mapsto\mathcal{A}_j$. Then, the state transfers to a next state according to the state transition dynamics $\mathcal{T}:\mathcal{S}\times\mathcal{A}_1\times\mathcal{A}_2\times ... \times \mathcal{A}_m\mapsto\Omega(\mathcal{S})$ where $\Omega(\mathcal{S})$ is a collection of probability distributions over $\mathcal{S}$. The environment returns an immediate reward to each agent based on a function of current state and all the agents' actions as $r_{j,t}:\mathcal{S}\times\mathcal{A}_1\times\mathcal{A}_2\times ... \times \mathcal{A}_m\mapsto\mathcal{R}$. The goal of each agent is to maximize its the total expected return $R_j=\sum_{t=1}^T\gamma^{t-1}r_{j,t}$ where $\gamma$ is a discount factor and $T$ is the time horizon. The detailed information can be found as below:

\begin{itemize}
	\item [$\mathcal{O}_j$:] The observation $o_{j,t}$ of agent $j$ at time step $t$ should in principle reflect the contract status, which mainly includes the following three parts: first, the time information, which tells the agent the current stage of the impression allocation process; second, the contract information, including demand fulfillment status and speed of contract $j$; third, the context information of other contracts' fulfillment status, which facilitates the global optimization.
	
	\item [$\mathcal{A}_j$:] At time step $t$, each agent $j$ delivers action $\delta_{j,t}\in \mathcal{A}_j$ to modify $\alpha_{j,t}$ to $\alpha_{j,t+1}$, typically taking the form of $\alpha_{j,t+1} = \min\{\alpha_{j,t} + \delta_{j,t}\cdot p_j,\ p_j\}$, where $p_j$ is the unit under-delivery penalty of contract $j$.
	
	
	\item [$r_{j,t}$:] Since the goal of impression allocation is to maximize the $outcome$ presented in Section \ref{section_pf}, the reward is defined as the resulting $outcome$ of the auction between time step $t$ and $t+1$. Specifically, let the impression set between time step $t$ and $t+1$ be $\mathcal{I}$, then $r_{j,t} = \Sigma_k (1-\Sigma_j x_{kj})\mathbbm{b}_{k2} + \Sigma_{kj} \lambda_jx_{kj}q_{kj} -  \mathbbm{1}_{t=T}\Sigma_j p_j y_j,\ k \in \mathcal{I}$.
	
	
	
	\item [$\mathcal{T}$:] We apply a model-free RL method to solve the impression allocation problem, so that the transition dynamics could not be explicitly modeled.
	
	\item [$\gamma$:] The reward discount factor $\gamma$ is set to 1 since the optimization goal of the impression allocation problem is to maximize the return regardless of time.
\end{itemize}

Usually, the optimal policies in a Markov game are difficult to learn due to some common challenges such as high model complexity, unstable return (\textit{i.e.}, $\Sigma_j\Sigma_t r_{j,t}$) caused by joint agents' updating in a sequential decision process. Fortunately, as for our problem, there is an important property that can simplify the policy searching process. Next, we first present this important property and then show the details of our method. 

\subsection{The Sub-problem in Impression Allocation} \label{section:subproblem}

In Theorem \ref{algo:prove}, we derive the optimal bidding function for the impression allocation problem. At each step, with contracts partially fulfilled, agents would face a sub-problem of impression allocation and are required to take actions based on the current state. We prove that the sub-problem can be formulated in the same form as \eqref{lp1} and the optimal action for each agent $j$ is shown by the following Theorem \ref{theorem2}.

\begin{theorem} \label{theorem2}
	For a sub-problem at each time step $t$, the optimal action sequence for each agent $j \in [1,..,m]$ is to modify its current $\alpha_{j,t}$ to the optimal ${\alpha}^*_{j,t}$, and keep it fixed for all its following time steps.
\end{theorem}

\begin{proof}
    At any time step $t$, since contracts may have already won several impressions, the demand of contracts for the subsequent impressions are refreshed. Specifically, let $e_j$ be the impression that contract $j$ has won, $\mathcal{I}$ be the remaining impression set at time step $t$. For impression $i\in \mathcal{I}$, the remaining demand of contract $j$ can be updated as $d'_j = d_j - e_j$, the impression under-delivery amount of contract $j$ as $y'_j=d'_j-\Sigma_{ij} x_{ij}$. Then, the components of $outcome$ can be updated as: $R'_{GC} = \Sigma_j c_jd'_j - \Sigma_j p_jy'_j$, $R'_{RTB}=\Sigma_i (1-\Sigma_j x_{ij})\mathbbm{b}_{i2}$, and $Q'_{GC} =\Sigma_{ij} \lambda_jx_{ij}q_{ij}$. 
    Therefore, the sub-problem can be formulated as the following linear programming, which shares the same form of \eqref{lp1}:
    
	\begin{equation} \tag{LP2} \label{lp3}
		\begin{aligned}
			\underset{x_{ij},y'_j}{\text{maximize}}  & & {R'}_{GC} + {R'}_{RTB} + {Q'}_{GC} & & & \\
			\text{s.t.} & & \Sigma_{ij} x_{ij} + y_{j} = d'_j,\quad &&&\forall j,i \in \mathcal{I}, \nonumber\\
			& & \Sigma_j x_{ij} \le 1,\quad &&& \forall i\in \mathcal{I}, \nonumber\\
			& & x_{ij} \ge 0,\quad &&& \forall j,i\in \mathcal{I}, \nonumber\\
			& & y_{j} \ge 0,\quad &&& \forall j. \nonumber
		\end{aligned}    
	\end{equation}
	
	Similarly, the optimal bid can be derived as $b'_{ij} = \lambda_j\ q_{ij} + {\alpha}^{*}_{j,t}$, where ${\alpha}^{*}_{j,t}$ is the optimal parameter of sub-problem at time step $t$. Thus, for the sub-problem of impression allocation, the optimal action is to adjust the current parameter ${\alpha}_{j,t}$ of contract $j$ to ${\alpha}_{j,t}^*$. 
	After taking the optimal action, for any contract $j$, the $\alpha_{j,t}$ would be the optimal one. Thus, the following optimal actions would be keeping the parameters fixed until the end of the episode.
\end{proof}

\subsection{Multi-agent Reinforcement Learning to Impression Allocation (MARLIA)}


Based on what has be introduced above, in this section, we present our policy searching method, named Multi-Agent Reinforcement Learning to Impression Allocation (MARLIA).
Firstly, we apply an actor-critic reinforcement learning (RL) model \cite{sutton2018reinforcement} as the implementation of our method for its simplicity\footnote{Without loss of generality, other actor-critic RL models also can be applied in our scenario.}.
Secondly, to enhance the scalability and reduce the model complexity, we let all agents share a same model and be differentiated through context information included in each agent observation (\textit{e.g.}, other contracts' fulfillment status).
Thirdly, also most importantly, based on Theorem \ref{theorem2}, the learning processes of actors and critics are simplified, which can be interpreted from the following two aspects: 

\begin{itemize}[align=right,leftmargin=0.15in]
	\item At time step $t$, the mission of each agent $j$ is to make a single optimal decision which adjusts current $\alpha_{j,t}$ to ${{\alpha}^*_{j,t}}$ and keep it fixed for the following time steps rather than sequentially adjusting it, which significantly reduces the learning difficulty.
	\item The learning process of the critic $Q$ is simplified as minimizing the difference between $v$ and $Q(o_{j,t}, \delta_{j,t})$ where $v$ is exactly the \textit{outcome} produced by the fixed parameters $\alpha_{k,t},\ k\in[1,..,m],$ over the remaining impression set. Compared with common practice of updating $Q(o_{j,t}, \delta_{j,t})$ towards $r_{j,t}+\gamma Q(o_{j,t+1}, \pi_j(o_{j,t+1}))$ (more details see temporal-difference method \cite{sutton2018reinforcement}), $v$ more clearly indicates whether the action $\delta_{j,t}$ would leads to a better final \textit{outcome}. This makes the learning process of $Q$ easier and, in turn, $Q$ will boost the policy convergence.
\end{itemize}

To be more concrete, MARLIA is presented in Algo. \ref{algo_fast}.

\begin{algorithm}
	\caption{MARLIA}\label{algo_fast}
	Initialize a random process $\mathcal{N}$ for action exploration\;
	Initialize replay memory $\mathcal{M}$ with capacity $N$\;
	Initialize policy $\pi_\theta$ with weights $\theta$\;
	Initialize state action value function $Q_\eta$ with weights $\eta$\;
	Set batch size to $\text{BS}$\;
	
	\While{not convergent}{
		
		Set $\alpha_{j,1}=\alpha^*_{j,1} + \mathcal{N}_1$\;
		Each agent $j$ bids with $\alpha_{j,1}$ via Eq. \eqref{eq:optimal_f} between time step $1$ and time step $2$;
		
		Calculate reward $r_{j,1}$\;
		
		\For{$t = 2$ \KwTo $T$}{
			
			\For{agent $j = 1$ \KwTo $m$}{
				Observe state $o_{j,t}$\;
				Get action $\delta_{j,t}=\pi_\theta(o_{j,t})$+$\mathcal{N}_t$ \;
				Set $\alpha_{j,t}=\min\{\alpha_{j,t-1}+\delta_{j,t}\cdot p_j, p_j\}$\;
			}
			
			Calculate reward $r_{j,t}$\;
			Set $v=r_{j,t}$\;
			
			\For{$t' = t+1$ \KwTo $T$}{
				Impression allocation with $\alpha_{j,t}$ via Eq. \eqref{eq:optimal_f}\;
				Calculate reward $r_{j,t'}$\;
				Set $v=v+r_{j,t'}$\;
			}
			
			Store $(o_{j,t}, \delta_{j,t}, v)$ in $\mathcal{M}$\;
			
			Sample $\text{BS}$ $(o^k, \delta^k, v^k)$ tuples from $\mathcal{M}$\;
			Update $Q$ by minimizing the loss $\mathcal{L}(\eta)=\frac{1}{\text{BS}}\Sigma_k(v^k-Q_\eta(o^k, \delta^k))^2$ \;
			
			Update policy using sampled policy gradient:
			
			\resizebox{.73\linewidth}{!}{$\nabla_{\theta}J\approx \dfrac{1}{\text{BS}} \sum_k \nabla_{\theta} \pi(o^k)\nabla_{\delta}Q^\pi(o^k, \delta)|_{\delta=\pi(o^k)}$}\;
			
		}
		
	}
\end{algorithm}

\section{Experimental Evaluation}\label{section:exp}

In this section, we first introduce the experimental setup, then we compare MARLIA with existing methods and show its advantages in \textit{outcome} maximization, last we investigated the converging efficiency and scalability of MARLIA, which is critical in real-world industrial application. 

\subsection{Experimental Setup}
\subsubsection{Dataset}

The experiment datasets are from a large advertising platform. The datasets consists of two publishers, and for each publisher, the ad serving logs are provided on May 17th-19th and June 17th-19th in 2020, and the contract demands are provided on May 18th-19th and June 18th-19th in 2020. We make each adjacent days of data as a training and testing pair, \textit{i.e.}, the previous day of impression data is used for training under the demand of the latter day, and the latter day of impression data with its demand is used for testing. Therefore there are 4 training sets and 4 test sets for each publisher. In total, the data contains more than 36 millions impressions with the detailed information, including time, market price and predicted click-through rate\footnote{Without loss of generality, we consider the click-through rate as impression quality.} for each contract bidder. 

\subsubsection{Evaluation Metrics}

The goal of impression allocation is to maximize the overall \textit{outcome}. Based on the optimal impression allocation formulated by \eqref{lp1}, the theoretically optimal \textit{outcome} on the testing dataset can be obtained, denoted by $R^*$. Let $R$ be the actual \textit{outcome} of the applied policy. The ratio between $R$ and $R^*$, \textit{i.e.}, $R/R^*$, is a simple and effective metric to evaluate the policy. For MARL algorithms, the convergence efficiency is also critical for practical effectiveness. It can be measured by converging time.

\subsubsection{Implementation Details}\label{imp_details}

The agents take action in every 15 minutes, so the $T$ in Algo. \ref{algo_fast} is $\num{96}$. We use two fully connected neural network, each of them with 2 hidden layers and 32 nodes per layer, to implement the policy $\pi_\theta$ and state action value function $Q_\eta$, respectively. The mini-batch size is set to $\num{32}$ and the replay memory size is set to $\num{100000}$. The action range is set to $[-0.1, 0.1]$ and the action noise is implemented by a normal distribution generator with $\mu=0$ and $\sigma=0.05$. 
We set the learning rate of actor and critic to $\num{1E-5}$ and $\num{1E-3}$ respectively. 
In our case, the model selection criterion in training process is choosing the best $R/R^*$ model with initial $\alpha_{j}^*$s from the day before the testing date within 6 hours training time. 
For practical considerations, for each day, we randomly sample 10\% of the data for model training, and use  GNU Linear Programming Kit (GLPK) to solve the dual problem of \eqref{lp1} to estimate the $R^*$ and $\alpha_j^*$.
In the evaluation period, we use 100\% of the data for testing. All experiments are carried on a Macbook Pro with 2.6 GHz Intel Core i7 and 16 GB 2667 MHz DDR4.

\subsubsection{Compared Methods} 

\begin{enumerate} [align=right,leftmargin=0.15in]
	\item \textbf{Fixed Parameter (FP):} A method that fixes the $\alpha_j$ as the $\alpha_j^*$ from training data for each contract, and delivers bids according to Eq. \eqref{eq:optimal_f} in testing period.
	
	\item \textbf{MSVV:} A classical algorithm for online ad allocation \cite{mehta2005adwords}.
	We set the bid for each contract as $(p_j + {\lambda_j}q_{ij})\cdot (1 - e^{x_j - 1})$, where $x_j$ is the the fraction of the bidder's budget that has been spent so far, and the bid of RTB as $\mathbbm{b}_{i2}\cdot (1 - e^{-1})$.
	
	\item \textbf{PID Controller (PID):} A widely used technique in display advertising 
	\cite{yang2019bid}
	to fulfill contracts by even pacing. We adopt this technique to modify $\alpha_j$ in Eq. \eqref{eq:optimal_f} to satisfy contracts' demand.
	
	
	
	\item \textbf{MARLIA:} The MARL to Impression Allocation method we proposed in this paper. 
\end{enumerate}


\begin{table*}
	\caption{The datasets statistics. The \textit{difference} means the daily relative change between adjacent days.}
	\centering
	\scalebox{1.0}{
		\setlength\tabcolsep{2pt}
		\begin{tabular}{@{}c|ccccc|ccccc@{}}
			\toprule
			& \multicolumn{5}{c|}{publisher1} & \multicolumn{5}{c}{publisher2} \\
			\midrule
			Dates & Impressions & Contracts & Demands & \tabincell{c}{Impression Volume\\Difference} & \tabincell{c}{Market Price\\Difference} & Impressions & Contracts & Demands & \tabincell{c}{Impression Volume\\Difference} & \tabincell{c}{Market Price\\Difference} \\
			\midrule
			05/17 & 4.34M & - & - & - & - & 2.07M & - & - & - & - \\
			05/18 & 4.20M & 45 & 1.62M & -3.5\% & 4.8\% & 2.01M & 31 & 1.00M & -2.8\% & 4.0\% \\
			05/19 & 4.05M & 49 & 1.72M & -3.5\% & 4.3\% & 1.95M & 30 & 1.01M & -3.4\% & 3.0\% \\
			06/17 & 4.14M & - & - & - & - & 2.07M & - & - & - & -  \\
			06/18 & 3.90M & 124 & 2.06M & -5.7\% & 4.9\% & 1.89M & 72 & 1.20M & -5.2\% & 4.4\% \\
			06/19 & 3.91M & 126 & 2.21M & 0.1\% & 5.0\% & 1.91M & 79 & 1.41M & 1.2\% & 5.4\% \\
			\bottomrule
		\end{tabular}
	}
\label{tb:ds_stat}
\end{table*}

\begin{table}
	\caption{The $R/R^*$ on 8 testing datasets of FP, MSVV, PID and MARLIA.}
	\centering
	\scalebox{1.0}{
		\setlength\tabcolsep{2pt}
		\begin{tabular}{@{}c|cccc|cccc@{}}
			\toprule
			& \multicolumn{4}{c|}{publisher1} & \multicolumn{4}{c}{publisher2} \\
			\midrule
			Dates & FP & MSVV & PID & MARLIA & FP & MSVV & PID & MARLIA \\
			\midrule
			05/18 & 0.851 & 0.866 & 0.938 & \textbf{0.945} & 0.846 & 0.883 & 0.932 & \textbf{0.954} \\
			05/19 & 0.909 & 0.849 & 0.923 & \textbf{0.949} & 0.928 & 0.878 & 0.929 & \textbf{0.953} \\
			06/18 & 0.902 & 0.863 & 0.893 & \textbf{0.951} & 0.870 & 0.874 & 0.885 &  \textbf{0.950}\\
			06/19 & 0.950 & 0.871 & 0.949 & \textbf{0.974} & 0.919 & 0.878 & 0.926 & \textbf{0.964}  \\
			\hline
			Average & 0.903 & 0.862 & 0.926 & \textbf{0.955} & 0.891 & 0.878 & 0.918 & \textbf{0.955} \\
			\bottomrule
		\end{tabular}
	}
\label{tb:result}
\end{table}

\subsection{Evaluation Results}

\begin{figure}[htbp!]
	\centering
	\begin{subfigure}{0.38\textwidth}
		\includegraphics[width=\textwidth]{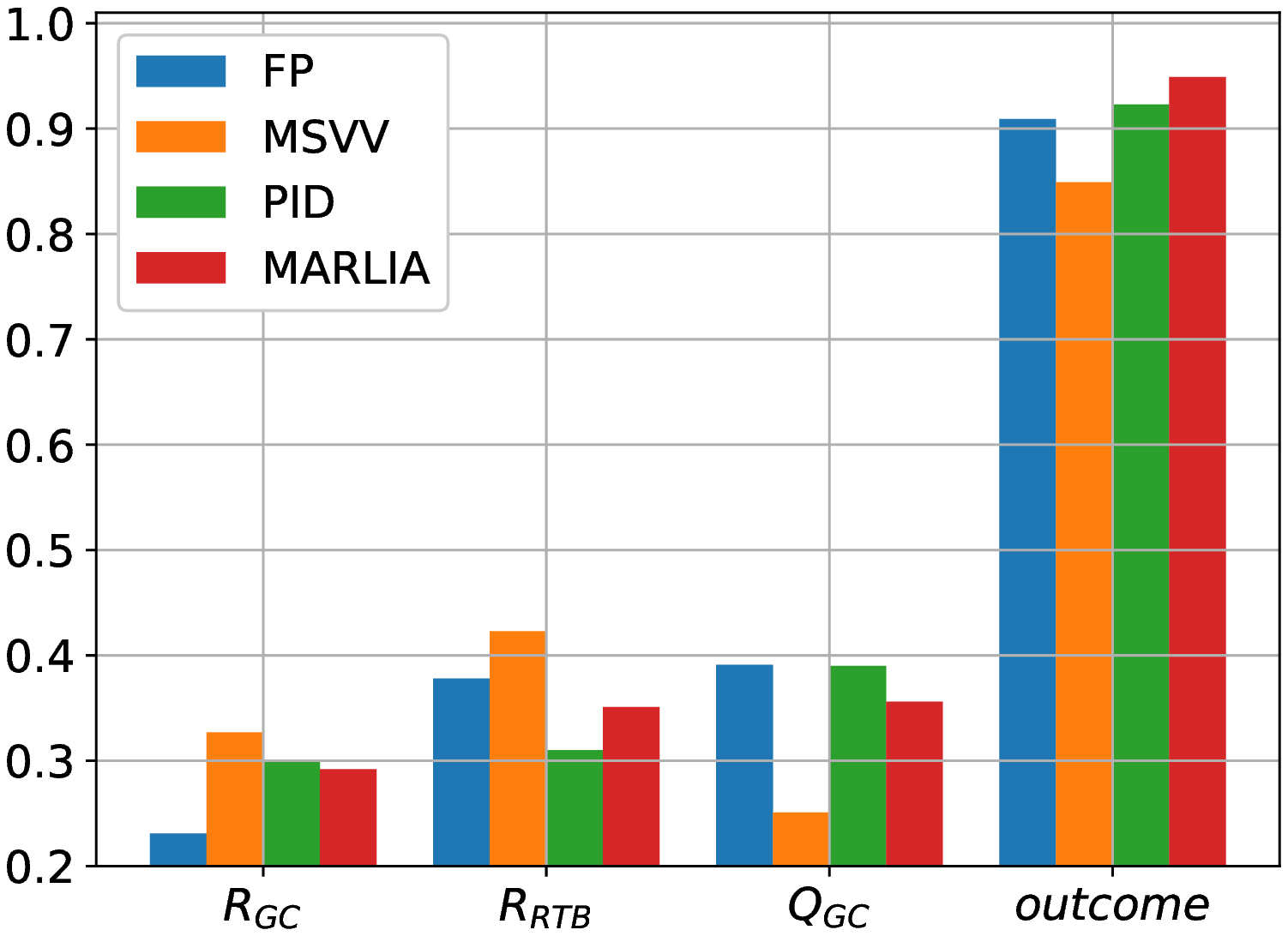}
		\caption{}
		\label{fg:methods_analysis}
	\end{subfigure}
	\hfill
	\begin{subfigure}{0.4\textwidth}
		\includegraphics[width=\textwidth]{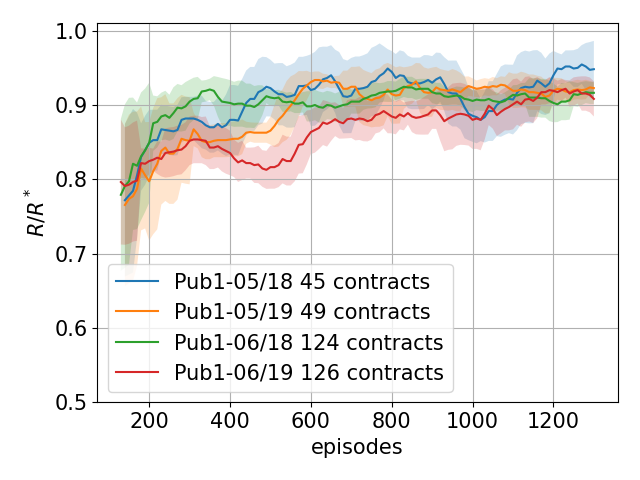}
		\caption{}
		\label{fg:convg_scala}
	\end{subfigure}
	\caption{
		\label{fig:NiceImage}%
		(a) The detailed results of all methods of publisher 1 on 2020/05/19, including $R_{GC}/R^*$, $R_{RTB}/R^*$, $Q_{GC}/R^*$ and $R/R^*$. (b) The training process of MARLIA on different datasets of Publisher 1.}
\end{figure}



We conduct experiments to compare the performance of FP, MSVV, PID and MARLIA. The initial parameter $\alpha_{j,1}$ for Eq. \eqref{eq:optimal_f} is set as the optimal one of the training data. As stated in Section 1, due to the unstable incoming traffic, the impression volume and market price in the testing data deviate from that of training data. To present the performances of different methods under different traffic patterns, the statistical information is summarized in Table \ref{tb:ds_stat}, and the experimental results based on testing dataset are summarized in Table \ref{tb:result}. We can see that MARLIA outperforms all methods, and the averaged improvements over FP, MSVV and PID (for publisher 1 / for publisher 2) are
$\num{5.8}$\%/$\num{7.2}$\%, $\num{10.8}$\%/$\num{8.8}$\% and $\num{3.1}$\%/$\num{4.0}$\% respectively.

To further investigate the behaviors of all methods, we go deep into the detailed results of all methods of publisher 1 on 2020/05/19. As shown by Fig. \ref{fg:methods_analysis}, we illustrate the $R_{GC}/R^*$, $R_{RTB}/R^*$, $Q_{GC}/R^*$ and $R/R^*$ of all methods.

\begin{itemize} [align=right,leftmargin=0.1in]
	\item \textbf{FP:} FP is a simple strategy and may achieve good results when  the environments between training and testing datasets is similar. However, when the bidding competition becomes fiercer or the impression volume has a shortage risk, FP would still bid with the fixed $\alpha_{j,t}$ at any time step $t$, which usually makes the bid too low to win sufficient impressions and results in large under-delivery penalty. It can be observed in Fig. \ref{fg:methods_analysis}, given the impression difference is -3.5\% and the market price difference is 4.3\%, although the $R_{RTB}$ and $Q_{GC}$ is slightly larger than MARLIA, the severe shortage of $R_{GC}$ affects the final \textit{outcome}.
	
	\item \textbf{MSVV:} MSVV takes the RTB revenue into consideration by trading off between the \textit{outcome} from contracts and that from RTB for each impression arrival. However, since it is unaware of the subsequent impression distribution, it may allocate impressions to contract in early phase of the whole allocation process, resulting in great loss of future $Q_{GC}$. Different from MSVV, MARLIA learns from the training data to make better decisions to maximize the overall \textit{outcome}. It can be observed in Fig. \ref{fg:methods_analysis}, due to such defect, MSVV gains low $Q_{GC}$ compared with MARLIA and finally results in low \textit{outcome}.
	
	\item \textbf{PID:} PID tries to adjust $\alpha_{j,t}$ to pace impression acquisition, however, it has two drawbacks in our application. First, it is critical for a PID strategy to obtain a proper target, which is the target impression volume to be allocated to the contract at each step. In practice, it is tricky and needs expert knowledge to be continuously optimized. Second, different from MARLIA, PID strategy is usually designed for an individual agent. It is hard to consider global impression allocation between different agents or balance RTB revenue and contract fulfillment. As shown by Fig. \ref{fg:methods_analysis}, PID tries to fulfill every contract and results in relatively low $R_{RTB}$, which affects its final \textit{outcome}.
	
	
	\item \textbf{MARLIA:} Compared with the above methods, MARLIA learns an $\alpha_{j,t}$ adjustment policy based on the current state to maximize the future \textit{outcome}. It can be found in Fig. \ref{fg:methods_analysis} that MARLIA does a good job in balancing different part of \textit{outcome}.
	
\end{itemize}

%

\subsection{Convergence Efficiency and Scalability}\label{scalability}


In application, FP, MSVV and PID have little training cost, while complicated methods such as MARL ones commonly suffered from convergence efficiency and scalability challenges. Besides, for an impression allocation task, it is of great necessity to accomplish the training process in an acceptable period of time. 
In order to show the converging efficiency and  scalability of MARL, the training process of MARLIA on the datasets of publisher 1 are illustrated by Fig. \ref{fg:convg_scala}. It can be seen that MARLIA will converge to a satisfying $R/R^*$ (e.g. greater than $\num{0.9}$) in $\num{1200}$ episodes ($\simeq\num{6}$ hours on a laptop) regardless of the contract number. Thus the proposed MARLIA is scalable and applicable in industrial scenario.

\section{Related Work}\label{relatedwork}
    Impression allocation for \textit{outcome} maximization is one of the most important issues to monetize traffic for a publisher. Some algorithms have been proposed to help publishers allocate impressions to contracts without consideration of RTB \cite{bharadwaj2010pricing,bharadwaj2012shale,zhang2017}. \cite{ghosh2009bidding} is the first work considering contracts as bidders to compete with RTB, but the goal is to maximize the representativeness of contract impressions. 
\cite{chen2014dynamic} proposes a revenue maximization strategy based on allocating and pricing the future contract impressions, which does not involve RTB and the efficiency is not optimized. 
To maximize the \textit{outcome}, \cite{balseiro2014yield} tries to learn a stochastic policy to deliver reserve price for each impression. However, for a repeated auction model, in any equilibrium, reserve price will make bidders in an ad exchange tend to shade their bids (not incentive compatible) \cite{carare2012reserve}. 
Besides, the \textit{contract first} strategy applied in \cite{balseiro2014yield} has also been proved not optimal in our experimental evaluations. 
Recently, a new strategy maximizing the total revenue is proposed in \cite{jauvion2018optimal}. However, the challenge of unstable traffic patterns is not discussed.



Impression allocation for outcome maximization is one of the most important issues to monetize traffic for a publisher. Some algorithms have been proposed to help publishers allocate impressions to contracts without consideration of RTB\cite{bharadwaj2010pricing,bharadwaj2012shale,zhang2017}. As RTB becomes increasingly important, how to maximize the profit considering both RTB and guaranteed contracts starts to be an open question. \cite{ghosh2009bidding} was the first work considering contracts as bidders to compete with RTB: the publisher is regarded as a bidder and the impression would be allocated to the bidder with the highest bidd price. However, the goal of \cite{ghosh2009bidding} was to maximize the representativeness of contract impressions from the advertisers' perspective. \cite{chen2014dynamic} proposed a revenue maximization strategy based on allocating and pricing the future contract impressions, but the allocation was determined in advance rather than through RTB. To maximize the outcome, \cite{balseiro2014yield} tries to learn a stochastic policy to deliver reserve price for each impression. However, it is known that, for a repeated auction model, in any equilibrium, reserve price will make bidders in an ad exchange tend to shade their bids (not incentive compatible) \cite{carare2012reserve}. Besides, the \textit{contract first} strategy applied in \cite{balseiro2014yield} has also been proved not optimal in our experimental evaluations. Recently, a new strategy maximizing the total revenue was proposed in \cite{jauvion2018optimal}. However, the challenge of unstable traffic patterns was not discussed or addressed.

Reinforcement learning (RL) has been applied to solve a wide range of problems. Specifically in the computational advertising domain \cite{yuan2012sequential, amin2012budget, jin2018real}, RL has been leveraged to optimize bidding strategies \cite{cai2017real, wu2018budget}. Traditional reinforcement learning approaches achieved a huge success in single agent settings, but delivered poor performance in MARL \cite{matignon2012independent}. One issue is that simply considering other agents as a part of the environment often breaks the convergence guarantee and makes the learning process unstable \cite{colby2015counterfactual}. Nash equilibrium algorithms have been proposed to address this problem \cite{de2012polynomial}. However, the computational complexity of directly solving Nash equilibrium confines such algorithms within handful agents and prohibits them from real-world applications. Following this direction, \cite{yang2018mean} tries to improve the scalability via leveraging action information from neighboring agents, while the concept of neighborhood is hard to be defined in the advertising context. 

\section{Conclusion}\label{conclusion}

In online display advertising, guaranteed contracts and real-time bidding (RTB) are two major ways to sell impressions for a publisher, especially for those large publishers with in-house RTB. In this paper, we proposed a strategy to maximize the \textit{outcome} of a publisher by allocating impressions between guaranteed contracts and RTB. We proposed the \textit{outcome} maximization strategy by deriving the optimal bidding function when contracts are treated as bidders. Meanwhile, the impression allocation strategy can also keep the incentive compatibility of RTB, making truthful bidding still the dominant strategy. In order to implement the strategy with the practical challenges such as unstable traffic patterns, we proposed an efficient MARL method, MARLIA, to adjust the critical parameter $\alpha_j$ for contract $j$. Meanwhile, based on the important property of the impression allocation problem, the learning efficiency is significantly boosted, and making the deploy of our method in industrial scenario is feasible. Experimental evaluations on large-scale real-world datasets demonstrate that MARLIA outperforms other baselines in \textit{outcome} maximization. Meanwhile, empirical studies also show that MARLIA converges quickly regardless of the contract amount, which is important in industrial application.

\bibliographystyle{IEEEtran}
\bibliography{IEEEfull}




\end{document}